\RequirePackage{fix-cm}
\documentclass[letterpaper]{article}     
\usepackage{graphicx}
 \usepackage{mathptmx}      
\usepackage[centertags]{amsmath}
\usepackage{amssymb}

\usepackage{amsthm}

\usepackage{xcolor}
\usepackage{hyperref}

\newcommand{\F}{\mathbb{F}}
\newcommand{\Z}{\mathbb{Z}}

\newcommand{\divides}{\ensuremath{\ | \ }}
\newcommand{\inset}[1]{\ensuremath{\left\{ #1 \right\} }}
\newcommand{\inbrace}[1]{\ensuremath{\left\{ #1 \right\} }}
\newcommand{\inparen}[1]{\ensuremath{\left( #1 \right) }}

\newcommand{\inbrak}[1]{\ensuremath{\left[ #1 \right] }}
\newcommand{\inangle}[1]{\ensuremath{\left\langle #1 \right\rangle }}
\newcommand{\codepow}[2]{\ensuremath{#1^{\inangle{#2}}}}
\newcommand{\mindist}{\ensuremath{d_{min}}}
\newcommand{\suchthat}{\ensuremath{~\middle|~}}
\newcommand{\ie}{\emph{i.e.}, }

\newcommand{\defined}{\ensuremath{\stackrel{\mathrm{def}}{=}}}
\newcommand{\wt}{\ensuremath{\operatorname{wt}}}

\newcommand{\spn}{\ensuremath{\operatorname{span}}}
\newcommand{\sh}{\ensuremath{\mathfrak{s}}}
\newcommand{\shi}[1]{\ensuremath{\sh^{(#1)}}}
\newcommand{\vv}[1]{\ensuremath{\vec{#1}}}

\theoremstyle{definition}
\newtheorem{defn}{Definition}[section]

\theoremstyle{plain}

\theoremstyle{remark}

\newtheorem{lemma}{Lemma}
\newtheorem{theorem}{Theorem}
\newtheorem{remark}{Remark}
\newtheorem{definition}{Definition}

\newcommand{\coeff}{\ensuremath{\operatorname{coeff}}}
\newcommand{\interp}{\ensuremath{\operatorname{poly}}}

\newcommand{\supp}{\ensuremath{\operatorname{supp}}}

\graphicspath{{./}}

\begin{document}

\title{Properties of Constacyclic Codes Under the Schur Product}

\author{Brett Hemenway Falk\thanks{\texttt{fbrett@cis.upenn.edu}. Department of Computer Science, University of Pennsylvania} \and Nadia Heninger\thanks{\texttt{nadiah@cis.upenn.edu}. Department of Computer Science, University of Pennsylvania} \and Michael Rudow\thanks{\texttt{mrudow@andrew.cmu.edu}. Computer Science Department, Carnegie Mellon University}}



\maketitle

\begin{abstract}
For a subspace $W$ of a vector space $V$ of dimension $n$, the Schur-product space $\codepow{W}{ k}$ for $k \in \mathbb{N}$ is defined to be the span of all vectors formed by the component-wise multiplication of $k$ vectors in $W$. 
It is well known that repeated applications of the Schur product to the subspace $W$ creates subspaces $W, \codepow{W}{2 },\codepow{W}{3}, \ldots$ whose dimensions are monotonically non-decreasing. 
However, quantifying the structure and growth of such spaces remains an important open problem with applications to cryptography and coding theory. 
This paper characterizes how increasing powers of constacyclic codes grow under the Schur product
and gives necessary and sufficient criteria for when powers of the code and or dimension of the code are invariant under the Schur product.

\end{abstract}

\section{Introduction}

This paper explores properties of constacyclic codes under the Schur-product (component-wise multiplication) operation.
For any code linear code, $C$, (\ie a linear subspace of $\F^n$) the Schur-product operation gives a natural way 
of constructing a new code
\[
	\codepow{C}{2} \defined \spn \inbrace{ \inparen{ c_1d_1,\ldots,c_nd_n) \suchthat (c_1,\ldots,c_n), (d_1,\ldots,d_n) \in C } }
\]
The behavior of codes $C$ under the Schur-product operation has many applications in coding theory and cryptography 
(See \cite{M12,R15,CCMZ15} for a surveys of known results).
The two main questions are how the dimension of the code grows, and how the minimum distance of the code shrinks under 
repeated applications of the Schur-product operation.
Recently, in \cite{C18}, Cascudo explores these two questions in the setting of cyclic codes under a single application of the Schur product. Moreover, Cascudo focuses on (1) determining how to represent the square of a cyclic code in a manner conducive to determining a tight minimum distance lower bound and (2) leveraging this representation to obtain families of cyclic codes whose squares have large dimension and minimum distance. 

In this work, we focus on \emph{constacyclic codes}, a widespread and important class of codes that includes 
cyclic codes, and negacyclic codes, both of which are widely used in coding theory and cryptography.

Cyclic error-correcting codes are codes (linear subspaces of $\F^n$) that are closed under the cyclic shift operation 
\[
	(c_1,c_2,\ldots,c_n) \mapsto (c_n,c_1,c_2,\ldots,c_{n-1})
\]
A cyclic code of length $n$ over a field $\F$ corresponds to an ideal in the ring $\F[x]/(x^n-1)$, where polynomials 
are identified as a vector of their coefficients.
A \emph{constacyclic} code is a code that is closed under the constant-shift map
\[
	(c_1,c_2,\ldots,c_n) \mapsto (ac_n,c_1,c_2,\ldots,c_{n-1})
\]
for some $a \in \F$.  When the length of the code, $n$, is co-prime to the characteristic of $\F$, then constacyclic codes correspond to ideals in the ring $\F[x]/(x^n-a)$.
See \cite{B68,ASR01} for a review of some of the basic properties of constacyclic codes.
The two most important types of constacyclic codes, are cyclic codes and \emph{negacyclic codes} (corresponding to $a=-1$, \ie ideals in the ring $\F[x]/(x^n+1)$).

The Schur-product operation arises naturally in many contexts throughout cryptography, including cryptanalyzing the McEliece Cryptosystem \cite{mceliece} and constructing ``multiplication-friendly" secret-sharing schemes \cite{CCCX09}.

The McEliece Cryptosystem has attracted significant interest from the cryptographic community since its introduction 40 years ago, and it is currently an active topic of research due to its perceived 
resistance to quantum attacks.  Over the years, many modifications have been proposed to the McEliece Cryptosystem, mostly with the aim of improving efficiency.
The original McEliece cryptosystem used binary Goppa codes, but many variants have been proposed using different codes that allow for smaller keys or more efficient encoding and decoding algorithms \cite{N86,S94,BL05,W06,berger2009,Wang16,dags17}.
Although the original McEliece Cryptosystem remains secure, most of these variants have been successfully cryptanalyzed \cite{MS07,FM08,W10,Faugère2010,COT14,FGOPT13,CGGOT14,MP17,CLT18,BC2018}. 
One of the key features in most of the successful cryptanalysis efforts has been that the proposed codes have \emph{small Schur-product dimension} which leads to key-recovery or indistinguishability attacks.
In particular, this lends credence to the idea that codes with small Schur-product dimension appear to be unsuitable for use in the McEliece framework. 

In addition to the McEliece Cryptosystem, cryptosystems based on the problem of Learning With Errors over Rings (Ring-LWE) have attracted much attention based on their resistance to quantum attacks and homomorphic properties \cite{LPR13,ABDGZ18}.
The security of modern Ring-LWE cryptosystems is closely related to the hardness of decoding certain types of \emph{negacyclic} codes, and thus a better characterization of the Schur-product dimension 
of negacyclic codes will provide a better understanding of the security of certain Ring-LWE based cryptosystems.

Although using codes with small Schur-product dimension appears to weaken the security of code-based cryptosystems like the McEliece Cryptosystem, 
these codes have many benefits that allow them to be used constructively in other areas of cryptography.
One of the most important constructive applications of codes with small Schur-product dimension is in the construction of 
``multiplication-friendly" secret-sharing protocols \cite{CCCX09}, which are the building blocks of efficient, secure multiparty computation (MPC) protocols 
\cite{CDM00,CDN15}.  Thus, identifying classes of codes with small Schur-product dimension immediately yields new candidates for efficient secret-sharing schemes and MPC protocols.

This paper focuses on identifying how the dimension of a constacyclic code grows under repeated applications of the Schur product.
In particular, we will give efficient algorithms that, using the generator polynomial of a constacyclic code, 
can efficiently compute the following:
\begin{enumerate}
	\item
		The maximum dimension the code will grow to under repeated applications of the Schur-product operation
	\item
		The generators for powers of the code after reaching the equilibrium dimension
	\item
		The minimum distance of the code after reaching the equilibrium dimension
	\item
		The criteria under which the code or powers of the code is invariant under powers of the Schur product
\end{enumerate}
In doing so, it will supplement the existing framework for the design and analysis of related cryptosystems.

\section{Preliminaries}

This section will highlight key definitions which will be used throughout the paper. These definitions are closely related to key properties that the theorems will prove, as well as standard building blocks that the proofs will use to do so. They are necessary to follow the proofs, underlying motivations, and consequences of the results.

\begin{defn}[Linear Code]

For any finite field, $\F$, and positive integer, $n$, a vector-subspace $C$ of $\F^n$ is called 
a \emph{(linear) code} of block-length $n$.  Vectors in $C$ are called \emph{codewords}.

\end{defn}

\begin{defn}[Schur Product of Vectors]

	Let $C$ be a code in $\F^n$ and let $c, d \in C$. 
	The \textit{Schur product of vectors} $c$ and $d$, denoted $c * d$, is the component-wise product of the codes 
	\[
		c * d \defined (c_1 \cdot d_1, \ldots, c_n \cdot d_n)
	\]
	For $i \ge 1$, we denote taking the Schur product of $c \in C$ with itself $i$ times as $\codepow{c}{i}$.
	We denote taking the Schur product of all vectors $v$ in a set $V$ as $\Pi_{v \in V} v$.
	
\end{defn}

This work will focus on the powers of codes under the Schur-product operation.
The Schur-power code (usually referred to simply as the \emph{power} of the code) is the 
smallest code that contains the Schur-powers of all its elements.

\begin{defn}[Powers of Codes]

	For a code, $C \in \F^n$, we denote the Schur-power code as
	\[
		\codepow{C}{1} = C
	\]
	and 
	\[
		\codepow{C}{i} = \spn\inparen{ \inset{ c*d \suchthat c \in C, d \in \codepow{C}{i-1}}  } \mbox{ for $i \ge 2$ }
	\]
\end{defn}

Throughout this work, we will identify polynomials of degree (at most) $n-1$ by their \emph{coefficient vectors}.

\begin{defn}[Polynomial Embeddings]
	Let $f(x) = x^n-a$ and $p(x) \in \F[x]  / f(x)$.  Then $p$ has a unique representation as
	\[
		p(x) = \sum_{i=0}^{n-1} c_i x^i
	\]
	We define 
	\[
		\coeff(p(x)) \defined (c_0, \ldots, c_{n-1})
	\]
	For any subset $C \subseteq \F^n$, we define 
	\[
		\interp(C) \defined \inset{ r(x) \suchthat \coeff(r(x)) \in C }.
	\]
	Similarly, for any vector $c = (c_0,\ldots,c_{n-1}) \in \F^n$, we define
	\[
		\interp(c) = c_0 + c_1 x + \cdots + c_{n-1} x^{n-1}
	\]
\end{defn}

We define the Schur power of a polynomial to be the Schur power of its coefficient vector.

\begin{defn}[Schur-powers of polynomials]
	If $h \in \F[x]/f(x)$, and $i \ge 0$, then 
	\[
		\codepow{h}{i} \defined \codepow{ \inparen{ \coeff(h(x)) } }{i} \in \F^n
	\]
\end{defn}

Thus $h^i$ corresponds to polynomial multiplication in the ring $\F[x]/f(x)$, whereas $\codepow{h}{i}$ corresponds to coordinate-wise powers of the coefficients of $h$.

\begin{defn}[Ideal Codes]
	Let $\F$ be a finite field, and $f(x) \in \F[x]$ a polynomial with $\deg(f) = n$.
	Then for any divisor $g(x)$ of $f(x)$, define the \emph{ideal code} generated by $g$ as
	\[
		C = \inset{ \coeff( g(x) \cdot h(x) \bmod f(x)) \suchthat h(x) \in \F[x]/f(x) } \subset \F^n
	\]
\end{defn}

This work will focus on a specific class of ideal codes known as \emph{constacyclic codes}.

\begin{defn}[Constacyclic Codes]
	Let $\F$ be a finite field, and $C$ an ideal code over $\F$ with modulus $f(x)$.	
	\begin{itemize}
		\item
			When $f(x) = x^n-a$, for some $a \in \F$, ideal codes over the ring $\F[x]/f(x)$ are  called \emph{constacyclic codes}. 
			Constacyclic codes are a subset of linear codes. Let $\ell$ denote the minimum natural number such that $a^\ell = 1$.
		\item
			When $f(x) = x^n-1$, ideals in $\F[x]/f(x)$ correspond to \emph{cyclic codes}.
		\item
			When $f(x) = x^n+1$, ideals in $\F[x]/f(x)$ correspond to \emph{negacyclic codes}.
	\end{itemize}
\end{defn}

\begin{defn}[Support of a vector]
	For a vector $c \in \F^n$, its \emph{support} is the set of indices of its nonzero entries
	\[
		\supp(c) = \inset{ i \suchthat c_i \ne 0 }
	\]
\end{defn}

\begin{defn}[Hamming weight]
	For a vector $c \in \F^n$, define
	\[
		\wt(c) = |\supp(c)|
	\]
\end{defn}

\begin{defn}[Minimum distance, minimum weight]
Let $C$ be a linear code over $\F^n$. Then the \emph{minimum distance} of $C$ is defined
\[
	\mindist(C) \defined \min_{c \ne d \in C} \wt(c - d).
\]
For linear codes, the minimum distance is equivalent to the minimum weight of a codeword in $C$
\[
	\mindist(C) = \min_{c \ne 0} \wt(c).
\]
\end{defn}

\begin{defn}[Generator Matrix ($G$)]
	Let $\F$ be a field and $C$ be a constacyclic code generated by polynomial $g(x)$ of degree $n-k$ which divides $x^n-a$. The \textit{generator matrix} for $C$ is defined as the $k \times n$ matrix where the $i$th row is equal to $\coeff(x^{i-1} \cdot g(x)$). Denote this matrix G. G is upper triangular because $\deg(g(x)) = n-k$.
	Thus if $g(x)= c_0 + c_1 x + \cdots + c_{n-k} x^{n-k}$, 
	\[
		G = \inbrak{ \begin{array}{ccccccccc}
			c_0 & c_1 & c_2 & \cdots & c_{n-k} & 0 & 0 & \cdots & 0 \\ 
			0 & c_0 & c_1 & \cdots & c_{n-k-1} & c_{n-k} & 0 & \cdots & 0 \\ 
			0 & 0 & c_0 & \cdots & c_{n-k-2} & c_{n-k-1} & c_{n-k} & \cdots & 0 \\ 
			\vdots & \ddots & \ddots & \ddots & \ddots & \ddots & \ddots & \ddots & \\
			0 & \cdots & 0 & 0 & c_0 & \cdots & c_{n-k-2} & c_{n-k-1} & c_{n-k} \\ 
			\end{array} }
	\]
\end{defn}

\begin{defn}[Standard Form Generator Matrix ($G'$)]
	Let $\F$ be a field and $C$ be a constacyclic code generated by polynomial $g(x)$ of degree $n-k$ which divides $x^n-a$. The \textit{standard form generator matrix} for $C$ is defined as the reduced row echelon form of the Generator Matrix $G$. The reduced Generator Matrix is  a $k \times n$ matrix whose leftmost $k \times k$ sub-matrix is $I_k$. Denote this matrix $G'$. Denote its $i$th row as $g_i$. Note that $g_k = \coeff(x^{k-1} \cdot g(x))$ because $G$ is upper triangular and Gaussian elimination on an upper triangular matrix doesn't change the final row.
	
\end{defn}

\begin{defn}[Shift]
	Let $\F$ be a finite field, and $C$ a constacyclic code over $\F$ generated by $g(x)$. Let $c = \coeff(p(x) \cdot g(x)) \in C$ for some polynomial $p(x)$. Then $\shi{i} c$ is defined to be
	\[
		\shi{i} c \defined \coeff( x^i \cdot p(x) g(x) )
	\]
	\ie if $c = (c_1,\ldots,c_{n})$ then 
	\[
		\sh c = (a \cdot c_n, c_1, c_2, \ldots, c_{n-1})
	\]
\end{defn}

\section{Prior Work}

There have been a number of recent results pertaining to properties of linear codes under the Schur-product operation \cite{M12,R15,CCMZ15}.
In this section, we highlight some of the relevant results.

If $C$ is a linear code, then the sequence of codes $\codepow{C}{1},\codepow{C}{2},\ldots$ has dimensions that are non-decreasing, 
and minimum distances that are non-increasing.
\begin{lemma}[\cite{R15}]
\label{R15-232}
For any linear code $C \subseteq \F^n$ and $z \ge 1$, 
\[
	\dim \inparen{ \codepow{C}{z+1} } \ge \dim \inparen{ \codepow{C}{z} }
\]
and
\[
	\mindist \inparen{ \codepow{C}{z+1} } \le \mindist \inparen{ \codepow{C}{z} }
\]
\end{lemma}

\begin{defn}[Hilbert Sequence \cite{R15}]
	\label{def:Hilbert}
	Let $C \subseteq \F^n$ be a linear code. The sequence of integers $\dim \inparen{ \codepow{C}{i} }$, for $i \ge 0$, is called the dimension sequence, or the \textit{Hilbert sequence}, of $C$. 
	The sequence of integers, $\mindist{ \inparen{\codepow{C}{i}} }$ for $i \ge 0$ is called the distance sequence of $C$.
\end{defn}

Since the dimension sequence is non-decreasing, and $\dim \inparen{\codepow{C}{i}} \le n$ there must exist a point at which the dimensions stop growing.

\begin{defn}[Castelnuovo-Mumford Regularity \cite{R15}]
	\label{def:CMR}

	The \textit{Castelnuovo-Mumford regularity} of a nonzero linear code $C \subseteq \F^n$ is the smallest integer $r = r(C) \ge 0$ such that 
	\[
		\dim \inparen{ \codepow{C}{r} } = \dim \inparen{ \codepow{C}{r+i} }
	\]
	for all $i \ge 0$.
\end{defn} 

The dimension sequence of a code is strictly increasing until it stabilizes, after which it never grows again.

\begin{lemma}[\cite{R15}]
\label{R15-233} 
Let $C$ be a linear code, $r(C)$ its Castelnuovo-Mumford Regularity, then for $z \in \inset{ 1,\ldots,r(C)-1}$
\[
	\dim \inparen{ \codepow{C}{z+1} } > \dim \inparen{ \codepow{C}{z} }
\]
\end{lemma}

The dimension sequence of a code no longer increases once the code is generated by a basis with disjoint support.

\begin{lemma}[\cite{R15}]
\label{R15-236} 
	For any linear code $C \subseteq \F^n, z \ge 0, z \in \mathbb{Z}$. 
	Then $z \ge r(C)$ if and only if 
	$$\dim \inparen{ \codepow{C}{z} } = \dim\inparen{ \codepow{C}{z+1} }$$
	which occurs if and only if $\codepow{C}{z}$ is generated by a basis of codewords with disjoint supports.  
\end{lemma}

\section{Basic known results about constacyclic codes}

In this section, we review some basic and well-known results about constacyclic codes.

We begin with a Lemma that determines a necessary and sufficient condition to show that a linear subspace is a constacyclic code.

\begin{lemma}[Closure under shifts]
	\label{lem:close_shift}
	A linear subspace $C \subset \F^n$ is a constacyclic code over the ring $F[x]/(x^n-a)$ if and only if 
	$C$ is closed under the shift operator 
	\begin{align*}
		\sh : \F^n &\rightarrow \F^n \\
				(c_0,\ldots,c_{n-1}) &\mapsto (a c_{n-1},c_0,\cdots,c_{n-2})
	\end{align*}
\end{lemma}

Lemma \ref{lem:basis} shows how to convert a generating polynomial for a code into a basis. 
\begin{lemma}[Basis of a constacyclic code]
	\label{lem:basis}
	For any constacyclic code $C$ of dimension $k$ over modulus $x^n-a$ and field $\F$, if $C$ is generated by a polynomial of minimal degree $g(x)$ of degree $n-k$, then $g(x) | x^n-a$ and $\{\coeff(x^i \cdot g(x)) \mid 0 \le i < k\}$ forms a basis for $C$.
\end{lemma}

The following lemma will be useful to determine whether a vector is nonzero.
\begin{lemma}[Consecutive zeros of a constacyclic code]
	\label{lem:consecutive}
	Let $C \subset \F^n$ be a constacyclic code of dimension $k$. Then $c \in C$ has $k$ consecutive zeros if and only if $c = 0^n$.
\end{lemma}

\begin{proof}
	In the first direction, if $c = 0^n$ then clearly it contains $n$ consecutive $0$s hence $k$ consecutive zeroes.
	
	In the second direction, suppose $c$ has $k$ consecutive zeroes starting in position $i$ when $c$ is considered as a one-indexed array. Let $c' = \shi{n-i+1} c$ and note that $c'$ has the same number of nonzero indices as $c$ and begins with $k$ consecutive zeroes.

	Let $G'$ be the standard-form generator matrix for $C$, with $g_1,\ldots,g_k$ the rows of $G'$.  Then $C$ is generated by $g_1,\ldots,g_k$. 
	The first $k$ positions of $c'$ are zero, hence 
	$$c' = \sum_{j=1}^{k} a_j g_j = \sum_{j=1}^k 0 \cdot g_j = 0^n.$$
	Since $c'$ is a shift of $c$, $c$ must also equal $0^n$. 

\end{proof}

The next lemma provides a lower bound on the minimum weight of a constacyclic code and the required support for a vector which achieves it.
\begin{lemma}[Minimum weight of a constacyclic code]
	\label{lem:min_weight}
	Any constacyclic code of length $n$ and dimension $k$ has minimum distance $\ge \frac{n}{k}$. Furthermore, if a code $c$ has weight $\frac{n}{k}$ and first nonzero position $p$, then $\supp(c) = \{p + z \cdot k, 0 \le z < \frac{n}{k}\}$. 
\end{lemma}

\begin{proof}
	By linearity of $C$, the code's minimum distance is the same as the minimum weight of a nonzero codeword. 
	Suppose towards contradiction that $c \neq 0^n \in C$ and $d = \wt(c) < \frac{n}{k}$. 
	Let $i$ be the first nonzero index of $c$ and let $c' = \shi{n-i+1} c$. Hence $c'$ starts with a nonzero index. Let $i_1 < i_2 < \ldots < i_d$ denote the indices of the nonzero coordinates of $c'$. Since $c'$ can never contain more than $k-1$ consecutive zeroes by Lemma \ref{lem:consecutive}, the largest position each nonzero coordinate can take is exactly $k$ positions more than the previous nonzero coordinate. Thus, $i_j \le 1 + (j-1)k <  n+1-k$. Hence, $c'$ ends in at least $k$ consecutive zeroes, and so $c' = 0^n=c$ by Lemma \ref{lem:consecutive}. Therefore, the original assumption is false; no $c \neq 0^n \in C$ can have  $\wt(c)< \frac{n}{k}$.
	
	In the case that $d = \wt(c) = \frac{n}{k}$, let $c', i_1,\ldots,i_d$ be defined similarly. Hence, $i_d \le 1 + (d-1)k \le  n-k+1$, where equality occurs if and only if each nonzero index occurs exactly $k$ positions after the previous nonzero index. As before, if $i_d < n-k+1$ there is a contradiction, thus $i_d$ must equal $n-k+1$. Hence, $\supp(c') = \{1 + z \cdot k, 0 \le z < \frac{n}{k}\}$. A simple shift back from $c'$ to $c$ completes the proof. 
\end{proof}

In the lemma below, we determine a sufficient condition for powers of a constacyclic code to be constacyclic. We will later see in Remark \ref{rem:non_chain} that a constacyclic code raised to an arbitrary power is not always constacyclic. 

\begin{lemma}[Schur product of a constacyclic code]
	\label{lem:chain}
	Let $C$ be a constacyclic code of length $n$, dimension $k$ and generator $g(x)$ over modulus $x^n-a$ and field $\F$. Let $\ell$ denote the minimum natural number such that $a^\ell = 1$. 

	Then for any $z \in \mathbb{N}$, $\codepow{C}{z \cdot \ell +1}$ is a constacyclic code of length $n$ over modulus $x^n-a$.

\end{lemma}

\begin{proof}
	By definition $\codepow{C}{t}$ is a linear code, for any $t \ge 0$, so
	to show that $\codepow{C}{z\cdot \ell + 1}$ corresponds to a constacyclic code, it suffices to show that $\codepow{C}{z \cdot \ell + 1}$ is closed under 
	the shift operator, $\sh$.
	Since $\codepow{C}{z \cdot \ell + 1}$ is spanned by vectors of the form $\prod_{i=1}^{z \cdot \ell +1} \vv{c}_i$, and the shift operator, $\sh$ is linear, 
	it suffices to show that $\sh \inparen{ \prod_{i=1}^{z \cdot \ell +1} \vv{c}_i } \in \codepow{C}{z \cdot \ell +1}$ whenever $\inset{ \vv{c}_i } \subset C$.
	Let $\vv{c}_i = (c_{1,i},\ldots,c_{n,i} )$, then 
	\begin{align*}
		\sh \inparen{ \prod_{i=1}^{z \cdot \ell+1} \vv{c}_i }  &= \sh \inparen{ \prod_{i=1}^{z \cdot \ell +1} c_{1,i}, \ldots, \prod_{i=1}^{z \cdot \ell +1} c_{n,i} } \\
												&= \inparen{ a \prod_{i=1}^{z \cdot \ell +1} c_{n,i}, \prod_{i=1}^{z \cdot \ell +1} c_{1,i}, \ldots, \prod_{i=1}^{z \cdot \ell +1} c_{n-1,i} } \\
												&= \inparen{ a^{z \cdot \ell +1} \prod_{i=1}^{z \cdot \ell +1} c_{n,i}, \prod_{i=1}^{z \cdot \ell +1} c_{1,i}, \ldots, \prod_{i=1}^{z \cdot \ell +1} c_{n-1,i} } \\
												&= \prod_{i=1}^{z \cdot \ell +1} \sh \vv{c}_i 
	\end{align*}
	Since $C$ is a constacyclic code, each $\sh \vv{c}_i \in C$, and so $\prod_{i=1}^{z \cdot \ell + 1} \sh \vv{c}_i \in \codepow{C}{z \cdot \ell + 1}$.

\end{proof}

\begin{remark}
	\label{rem:non_chain}
	It is not hard to find constacyclic codes that are not closed under the Schur-product operation, 
	\ie constacyclic codes $C$, such that $\codepow{C}{d}$ is not constacyclic.
	Note that in these cases, by Lemma \ref{lem:chain} it must be that $d \neq 1 \mod \ell$.
	Concretely, if $C$ is the constacyclic code generated by $g(x) = x^3+4$ over $\F_7/(x^6+2)$, then $\codepow{C}{2}$ is not constacyclic.
\end{remark}

As a result of Lemma \ref{lem:chain} and Remark \ref{rem:non_chain}, the following modifications to Definitions \ref{def:Hilbert} and \ref{def:CMR} are needed to reflect that only certain powers of constacyclic codes are constacyclic:

\begin{defn}[Constacyclic Castelnuovo-Mumford regularity]

Let $C$ be a constacyclic code generated by some $g(x)$ dividing modulus $f(x) = x^n-a$ over $\F$ with $a^{\ell} = 1$. Then the \textit{Constacyclic Castelnuovo-Mumford regularity}, $r'(C) = z$, is the unique $z \in \mathbb{Z}$ such that $z \ell + 1 \ge r(C) > (z-1) \ell + 1$.
\end{defn} 

\begin{defn}[Constacyclic Hilbert sequence]
Let $C \subseteq \F^n$ be a constacyclic code generated by some $g(x)$ dividing modulus $f(x) = x^n-a$ over $\F$ with $a^{\ell} = 1$. The \textit{constacyclic Hilbert sequence} of $C$ is defined as $\dim(C^{\langle i \ell + 1 \rangle}), i \ge 0$.
\end{defn}

In Lemma \ref{lem:generator}, we note that the constacyclic Hilbert sequence of a code can be computed efficiently, given only a basis for the code.

\begin{lemma}[Generator computation]
Given any constacyclic code $C$ where $\dim(C) = k$ and given at least $k$ linearly independent code words $c_1, \ldots, c_j$ for $j \ge k$, it is possible to determine $g(x)$, the generator of $C$ in $O(k^2 n)$ operations. Furthermore, given generator $g(x)$ of code $C$, it is possible to determine generator $q(x)$ of $\codepow{C}{z\ell+1}$ in $O(n^4 log(z \cdot \ell))$ operations.

\label{lem:generator}
\end{lemma}

\begin{proof}
We first determine how to obtain $g(x)$ given $c_1,\ldots, c_k$. We begin with forming $G'$ by applying Gaussian Elimination to the matrix with rows $c_1, \ldots, c_k.$ The $k$th row of $G'$, $g_k$ will be the final nonzero row and will be $\coeff (x^{k-1} g(x))$, thus completely determining $g(x)$. Gaussian Elimination of a $k \times n$ matrix of rank $k$ can be completed in $O(k^2 n)$ operations. 

Given a basis of $C$, it is possible to determine a basis of $\codepow{C}{z\ell+1}$ in $O(n^4 log(z \ell))$ operations. We first use the method of repeated squaring to determine the bases for $C,\codepow{C}{2}, \codepow{C}{4},\ldots, \codepow{C}{b}$ where $b = 2^{\lfloor log_2(z\ell+1) \rfloor}$. Then there is some $I \subseteq \{1,2,4,\ldots,b\}$ such that $\Pi_{i \in I} \codepow{C}{i} = \codepow{C}{z\ell+1}$. Since $|I| = O(log(z\ell+1))$, it suffices to show that given bases for $\codepow{C}{a},\codepow{C}{a'}$ it is possible to compute a basis for $\codepow{C}{a+a'}$ in $O(n^4)$ operations. 

Let $A,A'$ be the matrices representing the bases for $\codepow{C}{a}$ and $\codepow{C}{a'}$ respectively. Since $A$ and $A'$ generate codes in $\F^n$, the maximum rank these codes can have is $n$, thus $A$ and $A'$ have at most $n$ rows (and they both have exactly $n$ columns).
Then $\codepow{C}{a+a'}$ is the span of the Schur product of rows of $A$ with rows of $A'$. There are $n^2$ such rows, each of length $n$. Thus, a matrix $M$ of all such rows can be computed in $O(n^3)$ operations. Computing a basis for $\codepow{C}{a+a'}$ simply involves determining the reduced row echelon form of $M$. When applying Gaussian Elimination to $M$, for each pivot column it is possible to identify a row nonzero in that column in $O(n^2)$ operations. Then row reduction using the chosen row can be done in $O(n^3)$ operations. Since there are $O(n)$ pivot columns, determining the first $n$ rows of the reduced row echelon matrix can be done in $O(n^4)$ time. Finally, since the rank is at most $n$, we can simply ignore the remaining $n^2-n$ rows.

\end{proof}

\section{Extensions to quasi-twisted codes}

\begin{definition}[Quasi-twisted codes]
	Let $m, t \ge 1$, and set $n = mt \in \Z$.
	A code $C \subset \F^n$ is called ($t,a)$-\emph{quasi-twisted}, if there is an $t \ge 1$, and $a \in \F$, such that
	$C$ is closed under the map
	\begin{align*}
		T_a^\ell : \F^n &\rightarrow \F^n \\
			(c_0,\ldots,c_{n-1}) &\mapsto (a c_{n-t}, \ldots, a c_{n-1}, c_0, c_1, \ldots, c_{n-t-1} )
	\end{align*}
\end{definition}

Under this definition, constacyclic codes are $(1,a)$-quasi-twisted.
Quasi-twisted codes have been well-studied, \cite{LS01,Aydin2001,Radkova,L07,Aydin2007,Jia2012},
and have even been further generalized to \emph{multi-twisted} codes \cite{Aydin2017}
by allowing $t$ different constants in the shift operation.

It is not hard to check that for any non-trivial quasi-twisted code, it mus be that $t \divides n$, 
thus we define $m \defined n/t$.
With this definition, quasi-twisted codes correspond to $\F[x]/(x^m-a)$ submodules of $\inparen{ \F[x]/(x^m - a) }^t$.

Quasi-twisted codes are closely related to constacyclic codes, and if we permute the indices of a quasi-twisted code
\[
	(c_0,\ldots,c_{n-1}) \mapsto ((c_0,c_t,\ldots,c_{(m-1)t)},(c_1,c_{t+1},\ldots,c_{(m-1)t+1)},\ldots,(c_{t-1},c_{2t-1},\ldots,c_{(m-1)t+(t-1)}))
\]
then each block of $m$ digits forms a constacyclic code.  This motivates the following definition of a projection

\begin{align*}
	P_i : \F^n &\rightarrow \F^{m} \\
		(c_0,\ldots,c_n) &\mapsto (c_i,c_{t+i},c_{2t+i},\ldots,c_{(m-1)t + i} )
\end{align*}

Then if $C$ is a quasi-twisted code, $P_i(C)$ is a constacyclic code for $i = 0,\ldots,t-1$.
Since addition and scalar multiplication commute with these project operations, 
the Schur product also commutes with the projection.
The shift operation is also compatible with the projection, in particular
\[
	P_i \inparen{ \shi{t} c } = \sh\inparen{ P_i (c) }
\]
Thus, through an application of Lemma \ref{lem:close_shift}, one can inductively show that 
\[
	P_i\inparen{ \codepow{C}{z \ell + 1} } = \codepow{P_i(C)}{z\ell+1}
\]

We can therefore apply our results to each 
of the projected codes $P_i(\codepow{C}{z\ell+1})$ individually, 
to analyze the structure of $\codepow{C}{z \ell+1}$ when $C$ is a quasi-twisted code.

\section{Results}
\label{sec:results}
\subsection{Overview}
This section will introduce the novel results of this paper. The results use a specific component of a constacyclic code to determine key high level properties of the code's constacyclic Hilbert sequence. We denote this component the \textit{pattern polynomial} and define it as follows: 

\begin{defn}
\label{def:pattern_polynomial}

	Let $C \subset \F^n$ be a constacyclic code with modulus $f(x) = x^n - a$, generator $g(x) | f(x)$ with $g(0) = 1$.
	Let $p(x)$ be the highest degree polynomial, with degree $n-v$, $v | n$, and $p(0) = 1$, such that $ p(x) | g(x)$ and $\inset{\coeff{ (x^i p(x))} \suchthat 0 \le i < v }$ have disjoint support. 
	Then $p(x)$ is defined to be the \textit{pattern polynomial} of $g(x)$. 

\end{defn}

We observe that $p(x)$ always exists, as $p(x) = 1$ satisfies all requirements of a pattern polynomial except for when a higher degree pattern polynomial exists.

In Section \ref{sec:CHSPatPoly} we present Theorem \ref{thm:pattern}, in which we use the pattern polynomial to qualify the structure of a code's constacyclic Hilbert sequence. Specifically, we show the pattern polynomials of a code's constacyclic Hilbert sequence are simply the corresponding powers of the underlying code's pattern polynomial. Since the pattern polynomial divides each element of the code, one can deduce important properties of the codewords. For example, the support of every element of the code is closely related to the support of the pattern polynomial. Moreover, all vectors in the code, including the minimum weight vector, have weight which is divisible by the weight of the pattern polynomial. Furthermore, through use of Theorem \ref{thm:pattern_past_result} and Lemma \ref{lem:disj_supp_form}, the structure of the pattern polynomial can be further specified. Consequently, knowledge of the pattern polynomial uniquely determines a pair $(v,d)$ where $v |n $,$d$ is a unit, $d^{-\frac{n}{v}} = a$, and every codeword $c  = (c_1,\ldots, c_n)$ has the property $c_j = dc_{j-v}$ for all $j >v$. 

Later, in Section \ref{sec:CHSmaxdimgen} we use Theorem \ref{thm:pattern_past_result} to show that the generators of powers of a code in its constacyclic Hilbert sequence are simply the pattern polynomial raised to the corresponding powers for all terms at or after the constacyclic Castelnuovo-Mumford regularity. Given that a generator completely determines a code, the pattern polynomial for a base code is sufficient for understanding all such powers of the code in its constacyclic Hilbert sequence. Moreover, the structure of the pattern polynomial induces several additional key properties when it generates a code, including: (1) that the generator has minimum weight, (2) the minimum weight vector has weight $\frac{n}{k}$ where $k$ is the dimension of the code, and (3) the basis of shifts of the generator polynomial has disjoint support. These properties are explored further in Section \ref{sec:CHSmaxdimgen} through Lemma \ref{lem:disj_supp_form}. 

Next, in Section \ref{sec:dimInv}, we introduce Theorem \ref{thm:main_gen} which qualifies when the powers of a code are invariant under the Schur product for all terms in the constacyclic Hilbert sequence occurring at or after the constacyclic Castelnuovo-Mumford regularity. Such codes form perhaps the most intrinsically interesting class of codes studied in this work. In addition to the other properties described with Theorem \ref{thm:pattern} and Theorem \ref{thm:pattern_past_result}, Theorem \ref{thm:main_gen} ensures that in the $(v,d)$ pair described above also exhibits the property that $d^{\ell} = 1$.

Furthermore, in Section \ref{sec:idpat} we include Theorem \ref{thm:acquire_pat} where we describe an efficient algorithm to determine the pattern polynomial for a code given its generator. This is of practical importance, since, as previously described, the properties of the pattern polynomial are helpful in understanding a constacyclic code's behavior under the Schur product. Thus, in order to understand such properties of a given code, it suffices to compute the pattern polynomial.  

Finally, in Section \ref{sec:misc}, we provide a few additional results.

\subsection{Pattern polynomials of a constacyclic Hilbert sequence }
\label{sec:CHSPatPoly}
In Theorem \ref{thm:pattern} we show that the pattern polynomials for a constacyclic Hilbert sequence are given by powers of the pattern polynomial of the base constacyclic code. Consequently, identifying the pattern polynomial of the base constacyclic code is sufficient to identify all pattern polynomials in the sequence. 

\begin{theorem}[Pattern polynomials of a constacyclic Hilbert sequence ]
	Let $C \subset \F^n$ be a constacyclic code with modulus $f(x) = x^n - a$, dimension $k$, generator $g_0(x)$, and pattern polynomial $p_0(x)$. Consider the of the constacyclic Hilbert sequence $C, \ldots$, then for each $\codepow{C}{i \ell+1}$ the generator $g_i(x)$ has pattern polynomial $p_i(x)= \codepow{p_0(x)}{ i \ell + 1  }$.
	 \label{thm:pattern}	
\end{theorem}

\begin{proof}

We proceed with a proof by induction that $p_i(x) = \codepow{p_0(x)}{ i \ell + 1  }$.  In the base case, it is clear that $p_0(x) = \codepow{p_0(x)}{1}$ by definition. In the inductive hypothesis, we assume that for $0 \le i < \zeta$ that $p_{i}(x) = \codepow{p_{0}(x)}{ i \ell + 1  }$. In the inductive step $i = \zeta$, and we will show that $C^{\langle i \ell + 1 \rangle }$ has pattern polynomial $p_i(x) = \codepow{p_0(x)}{ i \ell +1  }$.

By Lemma \ref{lem:patDiv} (below) we have that $\codepow{p_0(x)}{i \ell+1} | g_i(x)$. 
Thus all that remains is to show that there is no higher degree polynomial $p'(x)$ that satisfies all other requirements for being a pattern polynomial of $g_i(x)$. 
Suppose towards contradiction that such a $p'(x)$ exists with $\deg(p'(x)) = n-v' > n-v = \deg(p_0(x))$. By Lemma \ref{lem:patAch}, we know that 
$$\codepow{p_0(x)}{(i \ell +1)(r'(C) \ell+1)} \in \codepow{C}{(i \ell+1)(r'(C)\ell+1)}.$$
Then by Lemma \ref{lem:patDiv} we also know that 
$$\codepow{p'(x)}{r'(C) \ell + 1} | \codepow{p_0(x)}{(i \ell +1)(r'(C) \ell+1)},$$ hence,
$$\codepow{p_0(x)}{(i \ell +1)(r'(C) \ell+1)} = \sum_{j=0}^{v'-1} c_j x^j \codepow{p'(x)}{r'(C) \ell + 1}.$$
This is a contradiction since 
$$\left|\supp\inparen{\codepow{p_0(x)}{(i \ell +1)(r'(C) \ell+1)}}\right| = \frac{n}{v}$$ while for some integer $z$ 
$$\left|\supp\inparen{\sum_{j=0}^{v'-1} c_j x^j \codepow{p'(x)}{r'(C) \ell + 1}}\right| = z\frac{n}{v'}.$$
We conclude that no such higher degree pattern polynomial $p'(x)$ exists. Therefore, $\codepow{p_0(x)}{i \ell+1}$ is the pattern polynomial for $\codepow{C}{i\ell+1}$.

\end{proof}

In Lemma \ref{lem:patDiv}, we show a relationship between the pattern polynomial of a base code and the minimum weight generator of any code in its constacyclic Hilbert sequence. This property is used in the proof of Theorem \ref{thm:pattern}.
\begin{lemma}
\label{lem:patDiv}
Let $C \subset \F^n$ be a constacyclic code with modulus $f(x) = x^n - a$, dimension $k$, generator $g_0(x)$, and pattern polynomial $p_0(x)$. Consider the constacyclic Hilbert sequence $C, \ldots$ with smallest degree generators $g_0(x), g_1(x), \ldots$. Then for each $i$, $p_0(x)^{i \ell + 1} | g_i(x)$ . 
\end{lemma}

\begin{proof}

Let $\codepow{C}{i \ell + 1} = (g_i(x))$. We will show that $\codepow{p_0(x)}{ i \ell+1}$ divides $g_i(x)$. By definition, 
$$g_i(x) = \sum_{j \in J} b_jr_j(x)$$
for units $b_j$ and $r_j(x) \in \codepow{C}{i\ell+1}$. Let $r(x) \in \{r_j(x) \mid j \in J\}$ arbitrarily. Then
\begin{equation}
\label{patDiv:r}
r(x) = \Pi_{j=1}^{i\ell+1} c_j(x)
\end{equation}
where each $c_j(x) \in (g_0(x))$. We will show that $p_0(x)^{i \ell+1}|r(x)$, and we then conclude $p_0(x)^{i \ell+1}|g_i(x)$. 

We note that $p_0(x) | g_0(x)$ and $g_0(x) | c_j(x)$, therefore, $p_0(x) | c_j(x)$. Hence we can rewrite 
$$c_j(x) = \sum_{z=0}^{n-1} d_{j,z} x^z p_0(x)$$
for constants $d_{j,z} \in \mathbb{F}$. We apply Lemma \ref{lem:cycleEquiv_disj_divides}, and conclude without loss of generality that 
\begin{equation}
\label{patDiv:c}
c_j(x) = \sum_{z=0}^{v-1} d_{j,z} x^z p_0(x).
\end{equation}
We substitute Equation \ref{patDiv:c} into Equation \ref{patDiv:r} to obtain
$$r(x) = \Pi_{j=1}^{i\ell+1} \sum_{z=0}^{v-1} d_{j,z} x^z p_0(x).$$
Then by associativity and commutativity, 
$$r(x)= \sum_{(z_1,\ldots,z_{i\ell +1}) \in [v-1]^{i \ell+1}} \Pi_{j=1}^{i \ell + 1} d_{j,z_j} x^{z_j} p_0(x).$$
Finally, since $(x^a p_0(x)) * (x^b p_0(x)) = 0$ unless $a=b$, we determine that
$$r(x) = \sum_{i=0}^{v-1} e_i x^i p_0^{i \ell +1}.$$
Thus, $p_0^{i \ell +1}$ divides each term in the linear combination whose summation is $g_i(x)$, and so $p_0^{i \ell + 1} | g_i(x)$.
\end{proof}

\subsection{Generators for codes of invariant dimension}
\label{sec:CHSmaxdimgen}
Next, in Theorem \ref{thm:pattern_past_result}, we determine the generator for all powers of a code in the constacyclic Hilbert sequence of maximal achieved dimension. Thus, Theorem \ref{thm:pattern_past_result} determines the behavior of a code under the Schur product once its dimension stops increasing. 

\begin{theorem}[Generators for codes of invariant dimension]
	Let $C \subset \F^n$ be a constacyclic code with modulus $f(x) = x^n - a$, generator $g(x)$, and pattern polynomial $p(x)$. Then for $z \ge r'(C)$, the generator $g_z(x)$ for $\codepow{C}{z\ell+1}$ is given by $\codepow{p(x)}{z \ell + 1}$. 
	 \label{thm:pattern_past_result}	
\end{theorem}

\begin{proof}
We provide a proof by induction on $z$. In the base case, we let $z = r'(C)$. We then apply Theorem \ref{thm:pattern} to determine that $\codepow{p(x)}{r'(C)\ell + 1}$ is the pattern polynomial for $\codepow{C}{r'(C)\ell+1}$. Furthermore, by definition of $r'(C)$ we note that
$$v = \dim(\codepow{C}{r'(C)\ell+1}) = \dim(\codepow{C}{(r'(C)\ell+1)(\ell+1)}).$$
We therefore apply Lemma \ref{R15-236} to determine that $\codepow{C}{r'(C)\ell+1}$ is generated by a basis of disjoint support. We then use Lemma \ref{lem:disj_supp_form} to determine that for some unit $u$, $v | n$, and $d^{-\frac{n}{v}} = a$, $\codepow{C}{r'(C)\ell+1}$ is generated 
$$g_{r'(C)}(x)= u \cdot \sum_{i=0}^{\frac{n}{v}-1} d^i x^{v \cdot i}.$$
By definition, $g_{r'(C)}(x)$ is its own pattern polynomial, so $g_{r'(C)}(x) = u \cdot \codepow{p_0(x)}{r'(C)\ell+1}$.

For the inductive step, we consider any $z > r'(C)$.  We note that 
$$\codepow{p_0(x)}{r'(C)\ell+1} * \codepow{g(x)}{(i-z)\ell} = u \cdot \codepow{p_0(x)}{i\ell+1}$$ for some unit $u$, hence, 
$$\codepow{p_0(x)}{i\ell+1} \in \codepow{C}{i\ell+1}.$$
Since $z > r'(C)$, we know that $v =  \dim(\codepow{C}{z\ell+1}) = n-\deg(p_0(x))$, so 
$$\codepow{C}{z\ell+1} = (B) = \inparen{\inset{\coeff\inparen{x^j \codepow{p_0(x)}{z\ell+1}} , 0 \le j < v}}.$$
Finally, an application of Lemma \ref{lem:disj_supp_form}, specifically the proof for $(5) \rightarrow (1)$, shows that $\codepow{p_0(x)}{z\ell+1}$ is a generator for $\codepow{C}{z\ell+1}.$
 
\end{proof}

In Lemma \ref{lem:patAch}, we show a relationship between the powers of a base code's pattern polynomial and powers of the code in the constacyclic Hilbert sequence for powers greater than or equal to the constacyclic Castelnuovo-Mumford regularity. Lemma \ref{lem:patAch} is crucial part of the proof of Theorem \ref{thm:pattern_past_result}.
\begin{lemma}
\label{lem:patAch}
Let $C \subset \F^n$ be a constacyclic code with modulus $f(x) = x^n - a$, dimension $k$, generator $g_0(x)$, and pattern polynomial $p_0(x)$ of degree $n-v$. Consider the constacyclic Hilbert sequence $C, \ldots$. Then for each $i$ such that $i \ell + 1 \ge  r'(C)\ell+1$, $\codepow{p_0(x)}{ i \ell +1} \in \codepow{C}{i \ell + 1}$. 
\end{lemma}

\begin{proof}
Clearly it suffices to prove 
$$\codepow{p_0(x)}{ r'(C)\ell + 1} \in \codepow{C}{ r'(C)\ell+1}$$
since then 
$$\codepow{p_0(x)}{ r'(C)\ell + 1}* \codepow{g_0(x)}{ (i-r'(C)) \ell} = \codepow{p_0(x)}{i \ell+1}.$$
We know that 
$$\dim(\codepow{C}{2(r'(C) \ell+1)}) = \dim(\codepow{C}{r'(C)\ell +1}.$$
Therefore, by Lemma \ref{R15-236}, $\codepow{C}{r'(C) \ell+1}$ is generated by a basis $B$ of vectors with disjoint support. By Lemma \ref{lem:disj_supp_form}, we know that the generator $g_{r'(C) \ell+1}(x)$ for $\codepow{C}{r'(C) \ell+1}$ is of the form
$$g_{r'(C) \ell+1}(x)= u \cdot \sum_{i=0}^{\frac{n}{v'}-1} d^i x^{v' \cdot i}.$$
Thus $g_{r'(C) \ell+1}(x)$ is its own pattern polynomial, and it suffices to show that 
$$g_{r'(C) \ell+1}(x) = u \codepow{p_0(x)}{r'(C) \ell+1}.$$

Suppose not, then there is some other pattern polynomial $p'(x) \in \codepow{C}{r'(C)\ell+1}$ with degree $n-v' > n-v$,$v' =\dim(\codepow{C}{r'(C) \ell+1})$ and $g_{r'(C) \ell+1}(x) = u p'(x)$. Let 
$$r_j(x) = u_j (x^j \codepow{g_0(x)}{\ell}) * g_{r'(C) \ell+1}(x)$$
where $u_j$ is a unit chosen to ensure that $r_j(0) = 1$ if $r_j(x) \neq 0$, if $r_j(x) = 0$, we simply define $u_j = 1$.

Since dimension doesn't increase and $g_{r'(C)\ell+1}$ is a minimum weight vector, 
$$\supp(r_j(x)) \in \{\emptyset, \supp(g_{r'(C)\ell+1}(x))\}.$$
For any $j \neq j'$ we will show $r_j(x) - r_{j'}(x) \in \{0, r_j(x),-r_{j'}(x)\}$. Clearly, if $r_j(x) = 0$ or $r_{j'}(x) = 0$ it holds. Thus we consider only when $r_j(x), r_{j'}(x) \neq 0$. In that case, 
$$\coeff(r_j(x))[1] = \coeff (r_{j'}(x))[1] = 1$$ and 
$$\forall m \in \{2,\ldots, v'\} g_{r'(C)\ell+1}(x)[m] =r_j(x)[m] = r_{j'}(x)[m] = 0.$$
Therefore, $\coeff (r_j(x)-r_{j'}(x))$ begins with at least $v'$ consecutive zeroes,  so $r_j(x) - r_{j'}(x) = 0$ by Lemma \ref{lem:consecutive}. 

In order to ensure such properties of the $r_j(x)$, it must be the case that
$$g_0(x) = u \sum_{j=0}^{v'-1} e_j p^*(x)$$
where $\deg(p^*(x)) = n-v'$ and $\supp(p^*(x)) = \supp(g_{r'(C) \ell+1}(x))$. Furthermore, $v' | n$ and without loss of generality we can chose the $e_j$ such that $p^*(0) = 1$. Therefore, $p^*(x)$ has all characteristics of the pattern polynomial except that it might not be of highest degree. Since $n-v' > n-v$, this violates that $p(x)$ was the pattern polynomial. Thus, the assumption $g_{r'(C) \ell+1}(x)$ is equal to some other pattern polynomial is false, 
$$g_{r'(C) \ell+1}(x) = u \codepow{p_0(x)}{r'(C) \ell +1}.$$

\end{proof}

In Lemma \ref{lem:disj_supp_form}, we will show that several different properties of a code are equivalent to the code having a basis with disjoint support. As determined in Lemma \ref{R15-236}, having a basis with disjoint support is necessary and sufficient to ensure that a code's dimension is invariant under the Schur product. Thus, it will be helpful to have a greater understanding of the structure of such a code. Lemma \ref{lem:disj_supp_form} is also used in the proof of Theorem \ref{thm:pattern_past_result}	.

\begin{lemma}[Disjoint support basis equivalence]
	\label{lem:disj_supp_form}
	Let $C \subset \F^n$ be a constacyclic code with modulus $f(x) = x^n - a$, generator polynomial $g(x)$ of degree $n-k$ such that $g(x) | f(x)$ and $k|n$ .
	Then the following are equivalent: 
	\begin{enumerate}
		\item
			$g(x)= u \cdot \sum_{i=0}^{\frac{n}{k}-1} d^i x^{k \cdot i}$ where $u$ is a unit and $d^{-\frac{n}{k}} = a$
		\item
			$G = uG'$ for some unit $u$
		\item 
			$B= \{\coeff (x^i \cdot g(x)), 0 \le i < k\}$ has disjoint support
		\item
			$C$ is spanned by a basis of disjoint support
		\item
			there is some vector $c \in C$ such that $\wt(c) = \frac{n}{k}$.
	\end{enumerate}
\end{lemma}

\begin{proof}

This proof will proceed by showing $(1) \rightarrow (2), (2) \rightarrow (3), (3)\rightarrow (4), (4) \rightarrow (5),$ and $(5) \rightarrow (1)$. 

We begin by showing that $(1) \rightarrow (2)$. Given (1), we know that G has rows with disjoint support and the first $k$ columns of $G$ are a diagonal matrix where each entry is $g(0)$. Therefore Gaussian Elimination would only multiply the matrix by $g(0)^{-1}$, leaving $G = u' G'$.

We will now show that $(2) \rightarrow (3)$. We note that if we take the first row of $G'$ it is equal to $u^{-1}g(x)$ for unit $u$. Then if we take $x^j u^{-1} g(x)$ for $0 \le j < k$, we see that 
$$\supp( \coeff(x^ju^{-1} g(x))) \cap \{1,\ldots, k\} = \{j+1\},$$
since the dimension is $k$ and so $G' = [I_K|M]$ for some $M \in \F^{k \times (n-k)}$. Given that the rows of $G'$ span $C$ and $\coeff (x^k g(x))[m] = 0$ for $1 <m \le k$, it is clear that $x^k g(x) = d g(x)$ for some nonzero $d \in \F$.  Consequently, $\supp (g(x)) = \supp (x^{ek} g(x))$ for any integer $e$. Hence, $j \in \supp (g(x))$ for $0 \le j < k$ if and only if $j+zk \in \supp (g(x))$ for all $0 \le z < \frac{n}{k}$. Since $deg(g(x)) = n-k$ we know $j \notin \supp (g(x))$ for $j \in [k-1] \setminus \{1\}$. Finally, we conclude $\supp (g(x)) = \{1 + zk, 0 \le z < \frac{n}{k}\}$, ensuring that $B$ has  disjoint support. 

Clearly, $(3) \rightarrow (4)$, as $\{ \coeff( x^j g(x)), 0 \le j < k\}$ is a basis of disjoint support that spans $C$. 

We now show $(4) \rightarrow (5)$. Suppose $C$ is spanned by a basis $B$ with disjoint support. For each $b \in B$ we know by Lemma \ref{lem:min_weight} that $\wt(b) \ge \frac{n}{k}$. Furthermore, since the vectors of $B$ have disjoint support we know $n \ge \wt(\sum_{b \in B} b) = \sum_{b \in B} \wt(b) \ge k \frac{n}{k} \ge n$. Thus $\wt(b) = \frac{n}{k}$ for each $b \in B$ to allow this. Hence, there is some $c \in C$ such that $\wt(c) = \frac{n}{k}$. 

We conclude by showing $(5) \rightarrow (1)$. Suppose $c \in C$ with $\wt(c) = \frac{n}{k}$, and without loss of generality assume $c[1] = 1$. Then by Lemma \ref{lem:min_weight}, $\supp(c) = \{1 + zk, 0 \le z < \frac{n}{k}\}$. Therefore, $B = \{s^i c, 0 \le i < k\}$ has disjoint support. This ensures that $B$ has dimension $k$, so it spans $C$. Furthermore, $\coeff(g(x)) = u c$ since $\coeff(g(x))[m] = 0$ for $n-k < m \le n$. We note that since $\supp (s^k c) = \supp (c)$ and $s^k c \in \spn(B)$, we can conclude $s^k c= e' c$ for some unit $e'$. Therefore, $ex^k g(x) = g(x)$ for some unit $e$. 
Hence,
$$g(x) = u \sum_{i=0}^{\frac{n}{k} -1} d_i x^{k i}$$
for unit $u$ and $d_0 = 1$. Then 
$$x^k g(x) -e g(x) = \inparen{\sum_{i=1}^{\frac{n}{k}-1} x^{ki}(ed_{i-1} -d_i) }+ (e d_{\frac{n}{k}-1}a-1) = 0.$$
This ensures that $d_{i+1} = e d_{i}$ and $e d_{\frac{n}{k}-1} = a^{-1}$. Consequently, $d_i = e^i$ and $d^{-\frac{n}{k}} = a$. 

\end{proof}

\subsection{Invariance under $\ell$-wise schur product}
\label{sec:dimInv}

In Theorem \ref{thm:main_gen}, we quantify how to identify codes whose powers are invariant under successive applications of the $\ell$-wise Schur product after the constacyclic Hilbert sequence reaches its maximum dimension. 
\begin{theorem}[Invariance under the $\ell$-wise Schur product]
  \label{thm:main_gen}
  
  	Let $C \subset \F^n$ be a constacyclic code with modulus $f(x) = x^n - a$, generator $g(x)$ and pattern polynomial 
  	$p(x) = \sum_{i=0}^{\frac{n}{v}-1} d^i x^{v i}.$
  	Then $\codepow{C}{r'(C)\ell+1} = \codepow{C}{z\ell+1}$ for all $z \ge r'(C)$ if and only if $d^{\ell} = 1$.

\end{theorem}

\begin{proof}
Suppose that $d^{\ell} = 1$. Then for $z \ge r'(C)$, by Theorem \ref{thm:pattern_past_result} we know that $\codepow{C}{z\ell+1} = (\codepow{p(x)}{r'(C)\ell+1}*\codepow{p(x)}{(z-r'(C))\ell})$. Furthermore, 
$$\codepow{p(x)}{(z-r'(C))\ell} = \sum_{i=0}^{\frac{n}{v}-1} (d^{\ell})^{(z-r'(C))i} x^{v i} = \sum_{i=0}^{\frac{n}{v}-1} x^{vi}.$$
Hence, $\codepow{C}{z\ell+1} =\codepow{C}{r'(C)\ell+1}$, as they are each generated by $(\codepow{p(x)}{r'(C)\ell+1})$.

Suppose $\codepow{C}{r'(C)\ell+1} = \codepow{C}{z\ell+1}$ for any $z \ge r'(C)$. Then by Theorem \ref{thm:pattern_past_result} we know $\codepow{C}{(r'(C)+1)\ell+1} = (\codepow{p(x)}{r'(C)\ell+1}*\codepow{p(x)}{\ell})$. Hence $\codepow{p(x)}{r'(C)\ell+1}, \codepow{p(x)}{(r'(C)+1)\ell+1}$ both belong to the code. Therefore, their difference is in the code. It is equal to 
$$\codepow{p(x)}{r'(C)\ell+1}-\codepow{p(x)}{r'(C)\ell+1}*\codepow{p(x)}{\ell} =  \sum_{i=0}^{\frac{n}{v}-1} d^{i (r'(C)\ell+1)}*(1- d^{\ell i}) x^{v i} = 0.$$

We note that it must equal zero since $p(0) = 1$ and $deg(p(x)) = n-v$. Therefore, the difference is zero in the final $v-1$ positions and first position. Thus, it has $v$ consecutive zeros, and by Lemma \ref{lem:consecutive} must equal 0. Hence, $ d^{i (r'(C)\ell+1)}*(1- d^{\ell i})  = 0$ for each $i$. Focusing on $i=1$, we see that $1= d^{\ell}$ as desired.

\end{proof}

\begin{remark}
As a motivation for Theorem \ref{thm:main_gen}, we observe that for a constacyclic code $C$ over modulus $x^n-a$, it is possible for $\dim(C) = \dim(\codepow{C}{\ell+1})$ while $C \neq \codepow{C}{\ell+1}$. Consider a code $C$ over $\F_5[x] / (x^4-1)$ and generator $g(x) = g(x) = x^3 + 2x^2+4x+3$. Then $\dim(C) = n-\deg(g(x)) = 1$, and $C = \spn([3,4,2,1])$. Furthermore, the basis of $\codepow{C}{2}$ is generated by $\spn([3,4,2,1]*[3,4,2,1])= \spn([4, 1,4,1])$. Thus $\dim\inparen{\codepow{C}{2}} = 1$. Yet, $[3,4,2,1] \notin \spn([4,1,4,1])$, so $C \neq \codepow{C}{2}$.
\end{remark}

\subsection{Determining the pattern polynomial}
\label{sec:idpat}

In Theorem \ref{thm:acquire_pat}, we will describe an efficient algorithm to compute the pattern polynomial for any constacyclic code.
\begin{theorem}[Determining the pattern polynomial]
Let $C \subset \F^n$ be a constacyclic code with modulus $f(x) = x^n - a$, generator $g(x)$, and dimension $k$. Let $w$ be the length of the input, where $w \ge \wt(g(x)) + \log(n) + \log(a)$ to include a description of $g(x),n,$ and $a$. It is possible to compute the pattern polynomial $p(x)$ in $O(w^2)$ time. After doing so, for $deg(p(x)) = n-v$, it is possible to compute $\{c_i\}$ such that $g(x) = \sum_{i=0}^{v-1} c_i x^i p(x)$ in $O(v)$ time.
\label{thm:acquire_pat}
\end{theorem}

\begin{proof}

The proof will proceed as follows. We begin by assuming without loss of generality that $g(0) = 1$. In $O(w)$ time, we will compute a candidate set $V$ of size $O(w)$ which contains $v$. We will then test each $v' \in V$ in $O(w)$ time to determine the whether there is some polynomial of degree $v'$ that satisfies all properties of the pattern polynomial, except perhaps being of highest degree. We take the highest degree such polynomial as the pattern polynomial. Afterwards, since $p(0) = g(0) = 1,$ we conclude $c_i = w[i+1]$ for $0 \le i < v.$

We know that a generator $g(x)$ has some pattern polynomial 
$$p(x) = \sum_{j=0}^{\frac{n}{v}-1} \alpha^i x^{v i}$$
such that 
$$g(x) = \sum_{i=0}^{v-1} c_i x^i p(x)= \sum_{i=0}^{n-1} b_i x^i.$$
We observe that $c_i = b_i$, as $p(0) = g(0) = 1$.
Clearly, $c_0 = 1$ since $g(0) = 1$, thus either $v = n$ or $\coeff (g(x))[v+1] \neq 0$. Therefore, in $O(w)$ time, we can compute 
$$V = \{n\} \cup \{i-1, \mid i \in \supp(\coeff(g(x)) \setminus \{1\} \},$$
while ensuring that $v \in V$ and $|V| = O(w)$.

For each $v' \in V$, we test in $O(w)$ time whether there is some polynomial $p'(x)$ of degree $n-v'$ satisfying all requirements of the pattern polynomial except perhaps being highest degree. If $v' = n$, then $p'(x)=1$ clearly satisfies all requirements of the pattern polynomial except perhaps being of highest degree. Otherwise, $p'(x)$ is the pattern polynomial if and only if (1) $v'|n$, 
$$(2)\text{ }p'(x) = \sum_{j=0}^{\frac{n}{v'}-1} d^j x^{v' j}, d^{-\frac{n}{v'}}= a,$$
$$(3)\text{ }g(x) = \sum_{j=0}^{v'-1}b_jx^j p'(x),$$
and (4) $p'(x)$ is the highest degree polynomial for which all these requirements are satisfied.

We note that the form for (2) is equivalent to the requirement that $\inset{\coeff(x^i p'(x)) \suchthat 0 \le i < v'}$ have disjoint support by Lemma \ref{lem:disj_supp_form}. Clearly, for $p'(x)$ to satisfy these requirements, $d=\coeff (g(x))[v'+1]$. It is straightforward to check (1) and that $d^{- \frac{n}{v'}} = a$ as well as compute both $p'(x)$ and 
$$g'(x) = \sum_{j=0}^{v'-1}b_jx^j p'(x).$$
To check (3), we simply verify that the $w$ nonzero coefficients of $\coeff (g(x))$ match the corresponding coefficients of $\coeff (g'(x))$ and that  no other nonzero coefficients of $g'(x)$ exists. Finally, (4) is achieved by taking the highest degree $p'(x)$ satisfying requirements (1), (2), and (3). Once all $v' \in V$ have been checked, we must be left with $p'(x)$ satisfying (1),(2),(3), and (4).

\end{proof}

\subsection{Additional Results}
\label{sec:misc}

In Lemma \ref{lem:cycleEquiv_disj_divides}, we expand upon the results from Lemma \ref{lem:disj_supp_form}, and show two properties of a polynomial taking the form of a generator found in Lemma \ref{lem:disj_supp_form}. These properties are used in the proofs of other lemmas in this work.
\begin{lemma}
Let $v|n$ and $p(x)= u \cdot \sum_{i=0}^{\frac{n}{v}-1} d^i x^{v \cdot i}$ where $u$ is a unit and $d^{-\frac{n}{v}} = a$. Then  (1) $x^v p(x) = d^{-1} p(x)$ over modulus $f(x)$ and (2) $p(x) | x^n-a$.
\label{lem:cycleEquiv_disj_divides}
\end{lemma}

\begin{proof}
We first show (1). By definition, $x^vp(x) - d^{-1} p(x) = (\sum_{i=1}^{\frac{n}{v}-1} (d^{i-1} - d^{-1} d^i)x^v) + d^{\frac{n}{v}-1} a - d^{-1} = 0$. 

We now show (2). In order to show $p(x) | x^n-a$ we simply show that $p(x) \cdot (x^v \cdot a \cdot d -a) =x^n-a$.
Upon evaluation, we determine
$$p(x) \cdot (x^{v} \cdot a \cdot d - a) = (\sum_{j=0}^{\frac{n}{v} -1} d^{j+1} \cdot a \cdot x^{v (j+1)})  - (\sum_{j=0}^{\frac{n}{v} -1} d^j \cdot a \cdot x^{v j}).$$
Reindexing and shifting terms yields
$$ d^{\frac{n}{v}} \cdot a \cdot x^{v \cdot \frac{n}{v}} + (\sum_{j=1}^{\frac{n}{v}-1} (d^{j} \cdot a \cdot x^{v j})  -d^j \cdot a \cdot x^{v j}) - a x^0.$$
We simplify to
$$ d^{\frac{n}{v}} \cdot a \cdot x^{n} - a = x^{n} -a.$$
Thus, $p(x) \cdot (x^{v} \cdot a \cdot d - a) = x^n-a$.
\end{proof}

The constacyclic Hilbert sequence monotonically increases in dimension until the dimension stops increasing. We show in Lemma \ref{lem:cyc}  that when $\F$ is a finite field, after the dimension stops increasing, the constacyclic Hilbert sequence cycles between at most $|\F|$ different codes. Thus, we further qualify the behavior of such codes of sufficiently large product dimension.
\begin{lemma}[Cycles in the constacyclic Hilbert sequence]
\label{lem:cyc}
Let $C \subset \F^n$ for $\F$ a finite field of order $|F|$ be a constacyclic code with modulus $f(x) = x^n - a$, generator $g(x)$, and pattern $p(x)= \sum_{j=0}^{\frac{n}{v}-1} c_j x^{v j}$. Then $\codepow{C}{r'(C)\ell+1}, \codepow{C}{(r'(C)+1)\ell+1}, \ldots$ forms a cycle of length at most $|\F|$. 	
\end{lemma}

\begin{proof}
By Theorem \ref{thm:pattern_past_result}, $\codepow{C}{z\ell+1} = (\codepow{p(x)}{z\ell+1})$ for $z \ge r'(C)$.  Thus, 
$$\codepow{C}{(r'(C)\ell+1) + |\F| \ell} = \inparen{ \codepow{p(x)}{r'(C)\ell+1} * \codepow{p(x)}{|\F|\ell}}.$$

We note that 
$$\codepow{p(x)}{|\F|\ell}=\sum_{j=0}^{\frac{n}{v}-1} (c_j^{|\F|})^{\ell} x^{v j} = \sum_{j=0}^{\frac{n}{v}-1} x^{v j}.$$
Therefore, $\codepow{C}{(r'(C)\ell+1) + |\F| \ell}  = \codepow{C}{r'(C)\ell+1}$.
\end{proof}

In Lemma \ref{lem:occurances}, we observe a bijection between nontrivial cyclic codes of length $n$ invariant under the Schur product and factors of $n$.

\begin{lemma}[Bijection between factors of $n$ and invariant codes]
	\label{lem:occurances}
	Consider a cyclic code $C$ of length $n$. Aside from the trivial subspace $C = \{0^n\}$ generated by $g(x) = 0$, there is a bijection between subspaces $C \subset \F^n$ where $C^{\langle 2 \rangle } = C$ and factors of $n$. In this bijection, for each factor, $k$, of $n$ the corresponding code $C$ has dimension $k$.

\end{lemma}

\begin{proof}

Let $m(k) = \sum_{i=0}^{\frac{n}{k}-1} x^{k \cdot i}$ for $k | n$. We will show that $m$ is a bijection between factors of $n$ and generators of nontrivial codes invariant under the Schur product. 

For any factor $k$ of $n$, let $g(x) = m(k) = \sum_{i=0}^{\frac{n}{k}-1} x^{k \cdot i}$. $(x^k -1) g(x) = x^n-1$ by Lemma \ref{lem:cycleEquiv_disj_divides} and $\dim((g(x)) = k$. Furthermore,  $g(x)$ is its own pattern polynomial so $r'(C) = 0$ and, by Theorem \ref{thm:main_gen}, $\codepow{C}{2} = C$. Thus $m$ maps factors of $n$ to nontrivial generators of $C$'s of the form $\codepow{C}{2} = C$. In doing so, $m$ is injective because $\deg(m(k)) \neq \deg(m(k'))$ for $k' \neq k$. 

For any nontrivial $C$ generated by $g(x)$ such that $\codepow{C}{2} = C$, we know $r'(C) = 0$. Therefore, $g(x)$ is its own pattern polynomial. Furthermore, by Theorem \ref{thm:main_gen}, $g(x) = \sum_{i=0}^{\frac{n}{k}-1} d^i x^{ki}$ where $d^\ell = d^1 = 1$. So for any generator $g(x)$ such that that $C^{\langle 2 \rangle } = C,$ it is clear that $g(x) = m(k)$. Thus, $m$ is surjective.
Hence, $m$ is bijective as desired.

\end{proof}

\begin{remark}

We observe that a constacyclic Hilbert sequence can include a dimension increase of $1$. Consider a code $C$ over $\F_3 / x^6-1$ generated by $g(x) = g(x) = x^4+2x^3+x+2$. Then $\dim(C) =2$ and $C = (\{(2,1, 0, 2,1,0), (0,2,1,0,2,1) \})$. Furthermore, 
$$\codepow{C}{2} = (\{(2,1,0,2,1,0)*(2,1,0,2,1,0), (2,1,0,2,1,0)*(0,2,1,0,2,1)\})$$
$$= (\{(1,0,0,1,0,0), (0,1,0,0,1,0)\})= (\{(1,0,0,1,0,0)\}).$$
Therefore, $\codepow{C}{2} = (1 + x^3)$, so $\dim(C^{\langle 2 \rangle }) =\dim(C) +1$.

\end{remark}

\section{Hilbert sequences in practice}

To provide some intuition about how the dimensions of constacyclic codes behave in practice, 
we generated the Hilbert sequences for cyclic and negacyclic codes in the rings 
$\F_q[x]/(x^{50} \pm 1)$.  We chose small a small subset of primes 
$q \in \inset{ \operatorname{nextprime}(2^i) \suchthat i \in \inset{8,\ldots,16} }$, 
and restricted our plots to those primes where the ring $\F_q[x]/(x^{50} \pm 1)$ had 
at most 1000 generators that generated codes of rate less than $1/2$.

For all tests, the Hilbert sequence stabilized at length 5, and the fraction of generators 
that produced given Hilbert sequence are shown in Figures \ref{fig:negahilbert}, \ref{fig:cyclichilbert}.
The $x$-axis shows the dimensions of the codes, thus the label $5$-$12$-$17$-$18$-$18$ 
corresponds to constacyclic codes, $C$, with 
$\dim \inparen{ \codepow{C}{1} } = 5$,
$\dim \inparen{ \codepow{C}{2} } = 12$,
$\dim \inparen{ \codepow{C}{3} } = 17$,
$\dim \inparen{ \codepow{C}{4} } = 18$,
and $\dim \inparen{ \codepow{C}{5} } = 18$.

\begin{figure}[h]
\includegraphics[width=.9\textwidth]{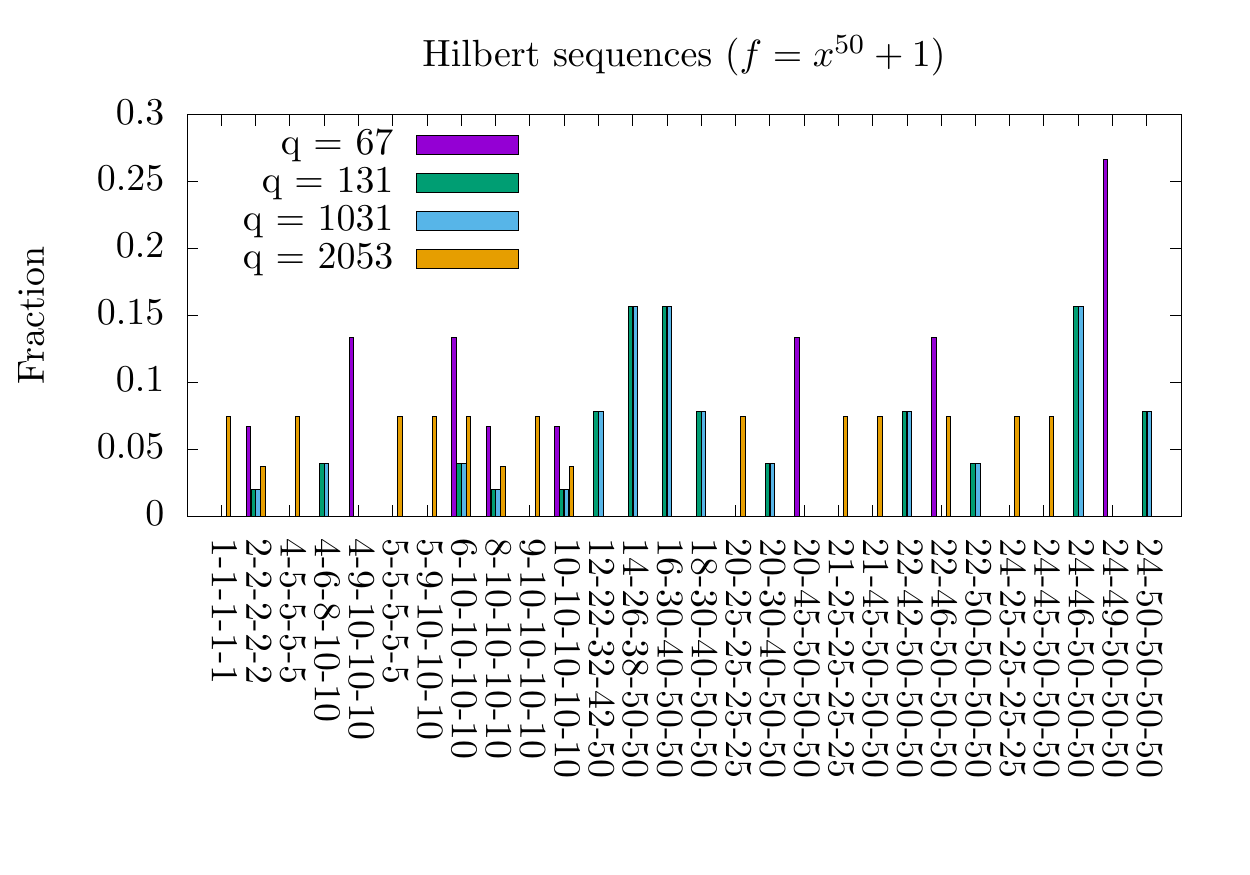}
\caption{Negacyclic Dimension Sequences. \label{fig:negahilbert}}
\end{figure}

\begin{figure}[h]
\includegraphics[width=.9\textwidth]{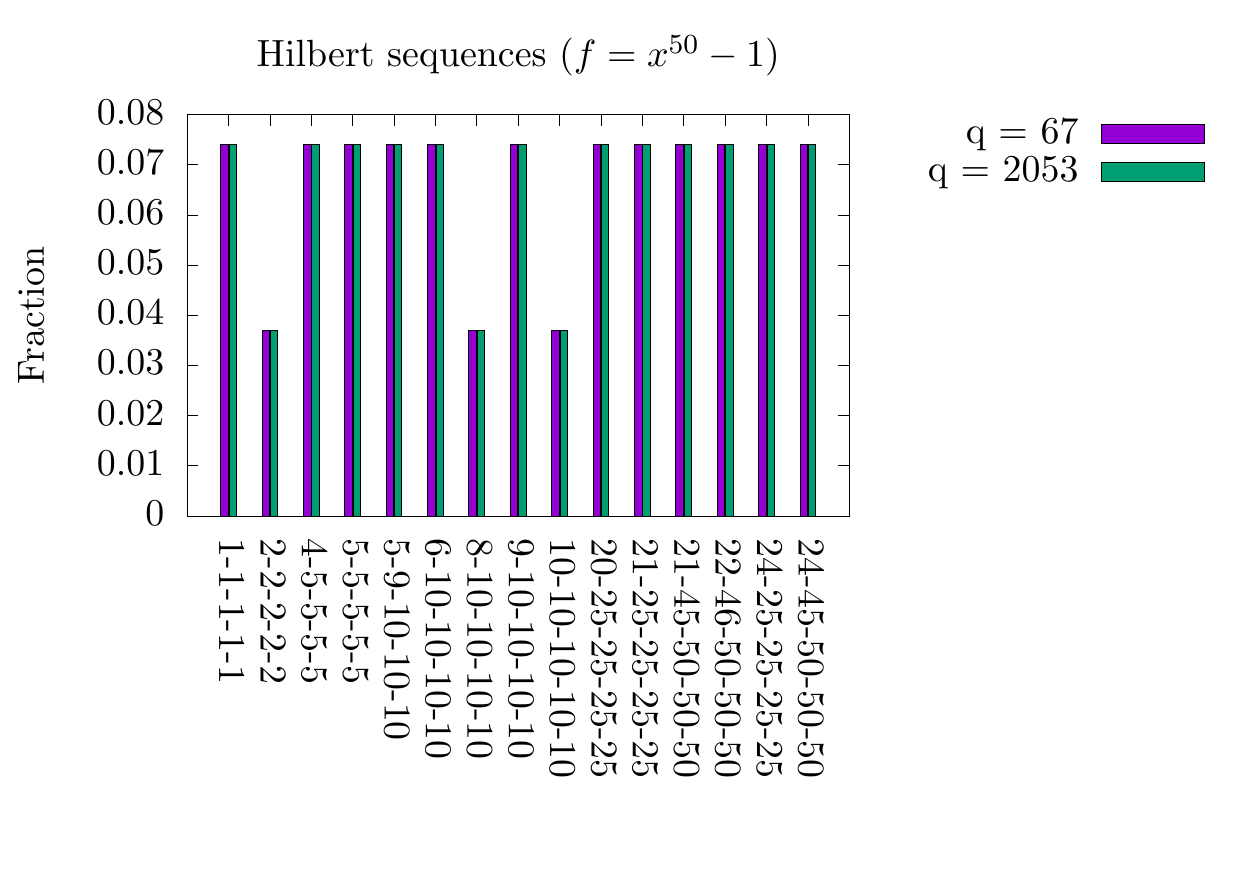}
\caption{Cyclic Hilbert Sequences.  There were only 27 generator polynomials that generated codes of rate less than $1/2$, and every Hilbert sequence appeared either once or twice.   \label{fig:cyclichilbert}}
\end{figure}

\section{Conclusion and Future Directions}
Overall, it is clear that constacyclic codes have structured growth under the Schur product. One can efficiently identify the dimension the code will grow to, the generators of the powers of the code past $r'(C)$, and when the code and or powers of the code are invariant under the Schur product. This inherent structure is important to consider when developing and or performing cryptanalysis of cryptosystems involving constacyclic or related codes to avoid inducing or to exploit vulnerabilities related to such properties.

The future directions of this work are twofold: improving the results proven in this paper and extending the results to other areas. In the first category, can the complexity bound for acquiring the pattern polynomial be improved from $O(w^2)$? Can $r'(C)$ be computed in time faster then $O(n^4 log(n))$? Can the time used to acquire generators of powers of $C$ be written in terms of the input length rather than $n$, perhaps by using sparse matrix operations? Can the chain of generators for $C,C^{\langle \ell+1 \rangle}, C^{\langle 2 \ell + 1 \rangle}, \ldots, C^{\langle (r'(C)-1) \ell + 1 \rangle}$ be obtained more efficiently through using properties of the pattern polynomial? In the second category, can any of these results be extended to the context of cyclic lattices? Can the techniques used in this paper be modified to apply to similar yet different results in cyclic lattices? Micciancio and Regev \cite{Micciancio08lattice-basedcryptography} wondered if one could safely use cyclic lattices in LWE-based cryptosystems to improve efficiency. Can such extensions be used to justify or preclude doing so? Can properties of constacyclic codes under the Schur product be used to design and or break future cryptosystems which use constacyclic codes?

\section{Acknowledgments}
This work was supported by the National Science Foundation under grants no. CNS-1651344 and CNS-1513671.



\bibliography{the_bib}
\bibliographystyle{spphys}       


\end{document}